\documentclass{ptephy_v1}

\usepackage[normalem]{ulem}

\theoremstyle{plain}
\newtheorem{df}{Definition}

\newtheorem{thm}{Theorem}

\newtheorem{cor}{Corollary}

\begin{document}

\title{Loosely trapped surface and dynamically transversely trapping surface in Einstein-Maxwell system}


\author{Kangjae Lee${}^1$}
\author{Tetsuya Shiromizu${}^{1,2}$}
\author{Hirotaka Yoshino${}^3$}
\author{Keisuke Izumi${}^{2,1}$}
\author{Yoshimune Tomikawa${}^4$}

\affil{${}^1$Department of Mathematics, Nagoya University, Nagoya 464-8602, Japan}
\affil{${}^2$Kobayashi-Maskawa Institute, Nagoya University, Nagoya 464-8602, Japan}
\affil{${}^3$Advanced Mathematical Institute, Osaka City University, Osaka 558-8585, Japan}
\affil{${}^4$Faculty of Economics, Matsuyama University, Matsuyama 790-8578, Japan}


\begin{abstract}
We study the properties of 
the loosely trapped surface (LTS) and the dynamically 
transversely trapping surface (DTTS) in Einstein-Maxwell systems.
These concepts of surfaces were proposed
by the four of the present authors in order to
characterize strong gravity regions.
We prove the Penrose-like inequalities for the area
of LTSs/DTTSs.
Interestingly, although the naively expected upper bound for the area
is that of the photon sphere of a Reissner-Nordstr\"om black hole
with the same mass and charge, 
the obtained inequalities include corrections
represented by the energy density or pressure/tension of electromagnetic fields.
Due to this correction, the Penrose-like inequality for the area of LTSs 
is tighter than the naively expected one. We also evaluate the
correction term numerically in the Majumdar-Papapetrou two-black-hole
spacetimes.
\end{abstract}

\subjectindex{E0, E31, A13}

\maketitle

%
\section{Introduction}
\label{section1}

Since a black hole creates a strong gravitational field,
there exists unstable circular orbits for photons.
For a spherically symmetric system, 
the collection of them makes a surface called a photon sphere. 
In the Schwarzschild spacetime, for example, the photon sphere exists at 
the surface $r=3m$, where $m$ is the Arnowitt-Deser-Misner (ADM) mass.
Furthermore, a generalized concept 
of a photon sphere,
which is called a photon surface, has been proposed \cite{Claudel:2000}. 
However, the definition of the photon surface 
requires a highly symmetric spacetime (e.g., \cite{Yoshino:2016}). 
Moreover, the existence of photon surfaces does not necessarily mean
that the gravitational field is strong there \cite{Claudel:2000}.
 
A photon sphere is directly related to observational phenomena. 
The quasinormal modes of black holes are basically determined
by the properties of photon spheres \cite{cardoso2009}. 
The black hole shadow, whose direct picture has been taken
by the recent radio observations of
the Event Horizon Telescope Collaboration \cite{akiyama},
is also determined by a photon sphere \cite{Virbhadra:1999}.

Motivated by the recent observations, the four of the present authors proposed 
concepts that characterizes strong gravity regions; a loosely trapped surface (LTS) 
\cite{shiromizu2017}, a transversely trapping surface (TTS) \cite{Yoshino:2017}, 
and a dynamically transversely trapping surface (DTTS) 
\cite{Yoshino:2020-1, Yoshino:2020-2}
(see also Ref. \cite{galtsov} for an extension of a TTS). 
For certain cases, we have proved inequalities analogous to the Penrose inequality~\cite{penrose1973}, 
that is, their areas are equal to or less than $4\pi (3m)^2$,
where $m$ is the ADM mass  
(see also Ref.~\cite{hod} for an earlier work for a photon sphere). 
The upper bound is realized for a photon sphere in
the Schwarzschild black hole, where 
$3m$ comes from the areal radius of unstable circular photon orbits. 
For an LTS, the application is restricted to a
spacelike hypersurfaces with a positive Ricci scalar. 
This restriction is natural because it is guaranteed by
the positivity of the energy density for maximally sliced initial data. 
For a DTTS, on the other hand, one requires
the non-positivity of the radial pressure on a DTTS in addition to
the positivity of the Ricci scalar. This requirement
of the non-positive pressure does not significantly
restrict the situation because 
the vacuum cases do work. However,  
it remains a little mystery why the non-positivity of
the radial pressure is required.

Therefore, in this paper, we shall discuss non-vacuum cases. As a first typical example, we will 
focus on Einstein-Maxwell systems. 
We adopt Jang's work \cite{jang1979} to show the inequality, 
that is, we will employ the method of 
the inverse mean curvature flow \cite{geroch, wald1977}
\footnote{The resulted inequality for the minimal surface in Ref.~\cite{jang1979} has a lower bound. 
  However, it is known that the lower bound is violated
  for multi-black-hole systems \cite{weinstein2004}. 
  The proof for multi-black-hole systems in the
  Einstein-Maxwell theory is given in Refs.~\cite{khuri2013,khuri2017}.}. 
The upper bound is expected to be given by
the area of the outermost photon sphere, i.e., 
the locus of the unstable circular orbits of photons 
in a spherically symmetric charged black hole spacetime, namely,
a Reissner-Nordstr\"om spacetime with the same mass and charge.
As a consequence, 
however, we see that the obtained inequalities depend on
(the part of) the energy density and the pressure/tension
of the electromagnetic 
fields which give corrections to the naively expected upper bound.
This is impressive because the Penrose inequality
for apparent horizons does not depend on such quantities.

The rest of this paper is organized as follows.
In Sect.~\ref{sec:EM-theory}, we will briefly describe
the Maxwell theory in a curved spacetime.
Some of the notations will be explained together. 
In Sect. \ref{sec:lts}, we will present the definition of the LTS  
and prove the Penrose-like inequality in the Einstein-Maxwell system. 
Then, in Sect. \ref{Sec:DTTS}, we will present the definition of the DTTS 
and prove the Penrose-like inequality in the Einstein-Maxwell system. 
In Sect.~\ref{Sec:numerical}, we will revisit the problem of DTTSs
in Majumdar-Papapetrou two-black-hole spacetimes, which
was studied in our previous paper \cite{Yoshino:2020-2},
from the viewpoint of the current study.  
We will examine the properties of the correction term
of the Penrose-like inequality through numerical calculations. 
The last section will give a summary and discussions.
In Appendix, we will shortly discuss 
the case of a TTS defined for static/stationary spacetimes. 
Note that we use following units in which 
the speed of light $c=1$, the Newtonian constant of gravitation $G=1$ 
and the Coulomb constant $1/(4\pi \varepsilon_0) =1$, where
$\varepsilon_0$ is the permittivity of vacuum.

%
%
\section{Einstein-Maxwell theory and setup}
\label{sec:EM-theory}

In this paper, 
we consider an asymptotically flat spacelike hypersurface $\Sigma$
in a four-dimensional spacetime with a metric $g_{ab}$. 
We suppose $n^a$ to be the future-directed unit normal to $\Sigma$
and the induced metric of $\Sigma$ is given by $\gamma_{ab}=g_{ab}+n_an_b$.
In $\Sigma$, we consider a two-dimensional closed surface,
an LTS (denoted by $S_0$) or a DTTS
(denoted by $\sigma_0$), with the induced metric
$h_{ab}$. The outward unit normal to that surface
is $r_a$ (and therefore, $\gamma_{ab}=h_{ab}+r_ar_b$).

We assume the presence of electromagnetic fields on $\Sigma$. 
The electromagnetic fields are specified by the anti-symmetric
tensor, $F_{ab}$, and its Hodge dual, ${}^{*}F_{ab}=(1/2)\epsilon_{abcd}F^{cd}$,
where $\epsilon_{abcd}$ is the Levi-Civita symbol
in a four-dimensional spacetime. The tensors $F_{ab}$
and ${}^{*}F_{ab}$ follow
Maxwell's equations,
\begin{equation}
\nabla_a{}^{*}F^{ab}=0, \qquad \nabla_aF^{ab}=-4\pi j^b,
\end{equation}
where $j^b$ is the
four-current vector.
The electric and magnetic fields, $E^a$
and $B^a$, are defined by
\begin{equation}
E_a:=F_{ab}n^b, \qquad B_a:=-{}^{*}F_{ab}n^b,
\end{equation}
respectively. 
Obviously, these fields are tangent to $\Sigma$ because $E_an^a=B_an^a=0$
holds. The electric charge density $\rho_{\rm e}$
and the electric current $(J_{\rm e})^a$ are
defined by
\begin{equation}
\rho_{\rm e}:= -j_an^a, \qquad (J_{\rm e})^a := {\gamma^a}_bj^b.
\end{equation}
In our paper, we require the charge density to vanish
(outside an LTS or a DTTS), i.e. 
$\rho_{\rm e}=0$, but do not necessarily require the
electric current $(J_{\rm e})^a$ to be zero.

The electric and magnetic fields satisfy Gauss' laws, $D_aE^a=4\pi \rho_{\rm e}$ and $D_aB^a=0$, 
where $D_a$ is the covariant derivative with respect to $\gamma_{ab}$. 
The total electric and magnetic charges are
\begin{equation}
4\pi q_{\rm e} :=\int_{S_\infty} E^ar_adA, \qquad 4\pi q_{\rm m} :=\int_{S_{\infty}}B^ar_adA,
\end{equation}
where $S_{\infty}$ is a sphere at spacelike infinity and $r_a$
is the outward unit normal to $S_{\infty}$.
It would be important to point out that $q_{\rm e}$ can have a nonzero value
even if the electric charge density is zero throughout the spacetime,
as one can understand by imagining
a spacelike hypersurface with an Einstein-Rosen bridge
and two asymptotically flat regions in 
the maximally extended Reissner-Nordstr\"om spacetime.
Similarly, although we assume the absence of magnetic
monopoles throughout the paper, the value of $q_{\rm m}$ 
can be nonzero. In a spherically symmetric spacetime, 
if both electric and magnetic fields
are present, the total squared charge defined by
\begin{equation}
q^2 := q_{\rm e}^2+q_{\rm m}^2
\end{equation}
appears in the spacetime metric of the Reissner-Nordstr\"om solution.
We handle the magnetic charge $q_{\rm m}$
in the following way. In Sects.~\ref{sec:lts} and \ref{Sec:DTTS},
we will assume $q_{\rm m}=0$ (and therefore, $q^2=q_{\rm e}^2$),
and derive the Penrose-like inequalities.
In the final section, we will discuss the modifications to
those inequalities when $q_{\rm m}$ is nonzero.

The energy-momentum tensor for the electromagnetic fields
is given by
\begin{equation}
  T_{ab}^{\rm (em)}
  = \frac{1}{4\pi}\left({F_a}^{c}F_{bc}-\frac14 g_{ab}F_{cd}F^{cd}\right).
  \label{EM-energy-momentum-tensor}
\end{equation}
In addition to $T_{ab}^{\rm (em)}$, we consider the presence of
ordinary matters whose energy-momentum tensor is $T_{ab}^{(m)}$. 
The total energy-momentum tensor is given by
$T_{ab}=T_{ab}^{\rm (em)}+T_{ab}^{(m)}$.
The relations of particular importance in this paper are
the energy density,
\begin{equation}
  8\pi \rho:=8\pi T_{ab}n^an^b =
  E^aE_a+B^aB_a +8\pi  \rho^{(m)},
  \label{energy-density-general}
\end{equation}
where $\rho^{(m)}:=T^{(m)}_{ab}n^an^b$, and the radial pressure,
\begin{equation}
  8\pi P_r:=8\pi T_{ab}r^ar^b =
  (E_aE_b+B_aB_b)h^{ab}-(E_ar^a)^2-(B_ar^a)^2+8\pi P_r^{(m)},
  \label{radial-pressure-general}
\end{equation}
where $P_r^{(m)}:=T^{(m)}_{ab}r^ar^b$.

%
%
\section{Loosely trapped surface in Einstein-Maxwell system}
\label{sec:lts}

In this section, we review the definition of an LTS
following Ref.~\cite{shiromizu2017},
and show the Penrose-like inequality for it 
in Einstein-Maxwell systems.

%
%

\subsection{Definition of an LTS}

The definition of an LTS is motivated by the following observation.
As an example, we consider a Reissner-Nordstr\"om spacetime. 
The metric is 
\begin{eqnarray}
ds^2=-f_{\rm RN}(r)dt^2+f^{-1}_{\rm RN}(r)dr^2+r^2 d\Omega_2^2,
\end{eqnarray}
where $f_{\rm RN}(r) :=1-2m/r+q^2/r^2$, and $m$ and $q$ are the ADM mass and
total charge, respectively. 
$d\Omega_2^2$ is the two-dimensional metric of the unit round sphere. 
From the behavior of a null geodesic, 
one can find unstable circular orbits of photons at 
\begin{eqnarray}
r=r_p:=\frac{3m+{\sqrt {9m^2-8q^2}}}{2}, \label{rp}
\end{eqnarray}
where we suppose $9m^2 \geq 8q^2$. Note that a photon sphere exists
even if the spacetime possesses a naked singularity at the center
for $(9/8)m^2\ge q^2>m^2$.

Now we define a similar concept to a photon sphere 
for general setups in terms of geometry. Here, we recall the fact that 
an apparent horizon is the minimal surface on  time-symmetric initial data. 
Therefore, one possibility to specify a strong gravity region
is to employ the mean curvature, that is,
the trace of the extrinsic curvature of two-dimensional surfaces. 
Therefore, we look at the mean curvature
for the Reissner-Nordstr\"om spacetime. 
It is easy to see that the mean curvature of an $r=$ constant surface on $t=$ constant hypersurface 
is given by
\begin{eqnarray}
k=\frac{2}{r}f^{1/2}_{\rm RN}(r).
\end{eqnarray}
From the first derivative of $k$ with respect to $r$, 
\begin{eqnarray}
  \frac{dk}{dr}=-\frac{2}{r^2}\left(1-\frac{3m}{r}
  +\frac{2q^2}{r^2}\right)f^{-1/2}_{\rm RN}(r),
\end{eqnarray}
we find that the maximum value of $k$ exists at $r=r_p$. This is 
exactly the same location with that of unstable circular orbits of photons. 
In the region between the event horizon and the photon sphere at $r=r_p$,
the mean curvature satisfies $k \geq 0$ and $dk/dr \geq 0$. 

From the above argument, one may adopt
the following definition of an LTS \cite{shiromizu2017}.

\begin{df} \label{def-LTS}
A loosely trapped surface (LTS), $S_0$, is defined as 
a compact two-surface in a spacelike hypersurface $\Sigma$, 
and has the mean curvature $k$ for the outward spacelike normal vector
such that $k|_{S_0}>0$ and $k'|_{S_0}\ge 0$, 
where $'$ is the derivative along the outward spacelike normal vector.
\end{df}

%
%

\subsection{Penrose-like inequality for an LTS}
\label{Sec:PenroseLike-inequality-LTS}

In this section, we present the inequality for the area of an LTS in
Einstein-Maxwell systems. Our theorem is as follows:

\begin{thm} \label{thm-LTS}
Let $\Sigma$ be an asymptotically flat spacelike hypersurface 
with the Ricci scalar ${}^{(3)}R \geq 2\left(E^aE_a+B^aB_a\right)+16\pi\rho^{(m)}$,
where $\rho^{(m)}$ is the non-negative energy density for other
matters.\footnote{From Eq.~\eqref{energy-density-general} and
  the Hamiltonian constraint of the Einstein equations,
  this condition is equivalent to $K_{ab}K^{ab}\ge K^2$, where $K_{ab}$
  is the extrinsic curvature of $\Sigma$. This condition is obviously
  satisfied by maximally sliced hypersurfaces, on which $K=0$ holds.}
We assume that $\Sigma$ is foliated by the inverse mean curvature flow, 
and a slice of the foliation parameterized by $y$, $S_y$, has topology $S^2$. 
We also suppose the electric charge density $\rho_{\rm e}$ to vanish
outside the LTS, $S_0$.
Then, the areal radius $r_0=(A_0/4\pi)^{1/2}$ of the LTS, $S_0$,
in $\Sigma$ satisfies the inequality 
\begin{eqnarray}
m \geq \frac{1}{3}\Bigl(1+\Phi_0^+ \Bigr)r_0+\frac{2q^2}{3r_0}, \label{ltsineq}
\end{eqnarray}
where $m$ is the ADM mass and $q$ is the total charge. $\Phi_0^+$ is defined by 
\begin{eqnarray}
  \Phi^+_0:= \frac{1}{8\pi} \int_{S_0}[(B_a r^a)^2+(E_aE_b+B_aB_b)h^{ab}]dA,
  \label{Phi-plus-0}
\end{eqnarray}
where $h_{ab}$ is the induced metric of $S_0$ 
and $ r^a$ is the outward-directed unit normal vector to $S_0$ in $\Sigma$. 
\end{thm}

\begin{proof}
On $\Sigma$, the derivative of the mean curvature $k$ along $r^a$ is given by
\begin{equation}
r^a D_a k \ = \ -\varphi^{-1}{\cal D}^2 \varphi -\frac{1}{2}{}^{(3)}R+\frac{1}{2}{}^{(2)}R  -\frac{1}{2}(k^2+k_{ab}k^{ab}), \label{deriv.mean.lts}
\end{equation}
where $D_a$ is the covariant derivative of $\Sigma$, 
${\cal D}_a$ is the covariant derivative of $S_0$, 
${}^{(2)}R$ is the Ricci scalar of $S_0$, 
${}^{(3)}R$ is the Ricci scalar of $\Sigma$,   
$k_{ab}$ is the extrinsic curvature of $S_0$ and
$\varphi$ is the lapse function for $y$, that is, $r_a=\varphi D_a y$. 
Then, the integration of Eq.~\eqref{deriv.mean.lts} over $S_0$ gives us 
\begin{eqnarray}
\frac{1}{2}\int_{S_0}{}^{(2)}RdA
&=&
\int_{S_0} \left[r^aD_ak+\varphi^{-2}\left(\cal{D}\varphi\right)^2
  +\frac{1}{2}{}^{(3)}R+\frac{1}{2}\tilde k_{ab}\tilde k^{ab}+\frac{3}{4}k^2 \right]
\nonumber\\
&\ge&
\frac{3}{4}\int_{S_0}\left[k^2+\frac{4}{3}(E^aE_a+B^aB_a)\right]dA
\nonumber\\
&=&
\frac{3}{4}\int_{S_0}\left[k^2+\frac{4}{3}(E_ar^a)^2\right]dA +8\pi \Phi^+_0,
\label{calculation-2DRicci-LTS}
\end{eqnarray}  
where $\tilde k_{ab}:=k_{ab}-(1/2)h_{ab}k$. Note that $\Phi^+_0 \geq 0$. 
Using the Gauss-Bonnet theorem and 
Cauchy-Schwarz inequality,  
we can derive the following inequality for the mean curvature
\begin{eqnarray}
\int_{S_0}k^2dA\le\frac{16\pi}{3}-\frac{4}{3}\frac{(4\pi q)^2}{A_0}-\frac{32}{3}\pi\Phi^+_0, \label{k-ineq}
\end{eqnarray}
where we used Gauss' law for the electric field, $\int_{S_0} E_a r^adA=\int_{S_\infty} E_a r^adA=4\pi q$. Here, $S_\infty$ denotes the two-sphere at
spacelike infinity.

Let us consider Geroch's quasilocal energy \cite{geroch, wald1977, jang1979}
\begin{eqnarray}
E(y):=\frac{A^{1/2}(y)}{64\pi^{3/2}}\int_{S_y}(2{}^{(2)}R-k^2)dA, 
\end{eqnarray}
where $A(y)$ is the area of $S_y$.
Here, we suppose that the surfaces $y=0$ and $y=\infty$
correspond to the LTS and a sphere at spacelike infinity, respectively. 
Under the inverse mean curvature flow
generated by the condition $k \varphi=1$,
the first derivative of $E(y)$ is computed as 
\begin{equation}
\frac{dE(y)}{dy} \ = \ \frac{A^{1/2}(y)}{64\pi^{3/2}}
\int_{S_y}\Bigl[ 2 \varphi^{-2}({\cal D} \varphi)^2
  +{}^{(3)}R+\tilde k_{ab}\tilde k^{ab} \Bigr] dA. 
\label{derivative-Geroch}
\end{equation}
Using ${}^{(3)}R \geq 2(E^aE_a+B^aB_a)$, we can derive  
\begin{eqnarray}
\frac{dE(y)}{dy} \geq \frac{A^{1/2}(y)}{32\pi^{3/2}}\left[\frac{(4\pi q)^2}{A(y)}
+8\pi \Phi^+_y\right], \label{monotonic}
\end{eqnarray} 
with the same procedure as the derivation
of the inequality of Eq.~\eqref{k-ineq},
where $8\pi \Phi_y^+=\int_{S_y}[(B_a r^a)^2+(E_aE_b+B_aB_b)h^{ab}]dA$. 

The integration of the inequality of Eq.~\eqref{monotonic} over $y$
in the range $0\le y<\infty$ implies us 
\begin{eqnarray}
  m & \geq &
  \frac{A_0^{1/2}}{4\pi^{1/2}}\left(1-\frac{1}{16\pi}\int_{S_0}k^2dA\right)
  +\frac{\pi^{1/2}q^2}{A_0^{1/2}}
  +\frac{1}{4\pi^{1/2}}\int^\infty_0 \Phi^+_yA^{1/2}(y)dy
  \nonumber \\
  & \geq &
  \frac{A_0^{1/2}}{6\pi^{1/2}}+\frac{4}{3}\frac{\pi^{1/2}q^2}{A_0^{1/2}}
+\frac{A_0^{1/2}}{6\pi^{1/2}} \Phi^+_0. 
\end{eqnarray}
In the above, we have used the well-known relation $A(y)=A_0\exp(y)$
that holds in the inverse mean curvature flow at the first step, and 
the non-negativity of $\Phi_y^+$ and
the inequality of Eq.~\eqref{k-ineq} in the second step.
Then, we find the inequality of Eq.~\eqref{ltsineq}. \\
\end{proof}

There are four remarks. First, the minimum value of the right-hand side of 
the inequality of Eq.~\eqref{ltsineq} implies  
\begin{eqnarray}
m \geq \frac{2{\sqrt {2}}}{3}{\sqrt {1+\Phi_0^+}}\ |q|. \label{mboundlts}
\end{eqnarray}
Setting $\Phi_0^+=0$, this inequality
is reduced to $m \geq \frac{2{\sqrt {2}}}{3}|q|$ which
corresponds to the condition for the existence of a photon sphere
in the Reissner-Nordstr\"om 
solution [see Eq.~\eqref{rp}]. 

Next, under the condition given by Eq.~\eqref{mboundlts},
the rearrangement of the inequality of Eq.~\eqref{ltsineq} 
gives us 
\begin{eqnarray}
4 \pi r_{{\rm LTS}-}^2 \leq  A_0 \leq 4 \pi r_{{\rm LTS}+}^2,  \label{cplilts}
\end{eqnarray}
where
\begin{eqnarray}
r_{{\rm LTS}\pm}:=\frac{3m \pm \sqrt{9m^2-8\left(1+\Phi^+_0\right)q^2}}{2\left(1+\Phi^+_0 \right)}.
\end{eqnarray}
This inequality must be interpreted carefully in the sense that
the lower bound would not hold in general.
In the case of the ordinary Riemannian Penrose inequality \cite{weinstein2004},
it has been pointed out that 
the lower bound is expected to be incorrect for multi-black-hole systems.
We consider this also may be the case for an LTS with multiple components.
This is not a contradiction to our proof since
an LTS with multiple components is out of the application of our theorem.
Therefore, in general,
we would have only the upper bound for the area,
\begin{eqnarray}
A_0 \leq 4 \pi r_{{\rm LTS}+}^2. 
\end{eqnarray}
Note that if we restrict out attention on an LTS with a single component,
the lower bound must hold true. The physical reason is as follows.
Let us consider a Reissner-Nordstr\"om spacetime
in the parameter region $(9/8)m^2> q^2>m^2$
which possesses both a naked singularity and two photon spheres.
The singularity of a Reissner-Nordstr\"om spacetime is known to be repulsive.
This is because 
the energy of electromagnetic
fields outside a small sphere near the singularity (e.g., in the sense of
the Komar integral) exceeds the
ADM mass $m$, and hence,
the gravitational field is generated by negative energy in that region.
The repulsive gravitational field is, of course,
not strong. 
This is the reason why an LTS with a single component
cannot exist in the vicinity
of a naked singularity with an electric charge and 
a lower bound exists for its area 
in the Einstein-Maxwell theory. 


The third remark is that there is no contribution from $\Phi_0^+$ in the 
Riemannian Penrose inequality for the Einstein-Maxwell system \cite{jang1979}, 
whereas in our theorem
it appears.
This is because the Riemannian Penrose inequality discusses 
the minimal surface with $k=0$, for which 
the inequality of Eq.~\eqref{k-ineq} is unnecessary.


Finally, the presence of 
$\Phi_0^+$ makes the inequality tighter 
than the case of $\Phi^+_0=0$.
From Eq.~\eqref{deriv.mean.lts}, one can 
see that the quantity $\Phi^+_0$ 
appears from (the part of) the energy density
in the Hamiltonian constraint. 
The electromagnetic energy density  
increases if $\Phi_0^+$ is turned on,
and through the relation given by Eq.~\eqref{deriv.mean.lts},
the positive energy density tends to make the
formation of an LTS more difficult. This means that
the area of 
$S_0$ will become smaller.

%
%

\section{Dynamically transversely trapping surface in Einstein-Maxwell system}
\label{Sec:DTTS}

In this section, we first explain the observation that
motivates the definition of a DTTS, and introduce 
the definition of
a DTTS. 
Then we prove an inequality for its area in Einstein-Maxwell systems.

%
%

\subsection{Definition of a DTTS}

The concept of a DTTS is inspired by the induced geometry of
a photon surface in a spherically symmetric spacetime.  
The photon surface is defined as a timelike hypersurface
$S$ such that any photon emitted tangentially to $S$ at an arbitrary 
point of $S$ remains in $S$ forever \cite{Claudel:2000}.
Let us consider spherically symmetric spacetimes with the metric,
\begin{eqnarray}
ds^2=-f(r)dt^2+f^{-1}(r)dr^2+r^2d\Omega_2^2.
\end{eqnarray}
Solving the null geodesic equations, 
we can find that $S$ satisfies
$dr/dt=\pm f (1-b^2f/r^2)^{1/2}$, where  $b$ is the impact parameter.
The induced metric of the photon surface $S$ is obtained as 
\begin{eqnarray}
ds^2|_S=-\alpha ^2 (r) dt^2+r^2d \Omega_2^2,
\end{eqnarray}
where $\alpha(r) :=bf(r)/r$. 
The mean curvature of $t=$ constant surface $\sigma_t$ in $S$ is given by 
\begin{eqnarray}
\bar k=\frac{2}{bf(r)}\frac{dr}{dt},
\end{eqnarray}
and the Lie derivative along $\bar n^a$ for $\bar k$ 
is computed as 
\begin{eqnarray}
  {}^{(3)}\bar{\mbox \pounds}_{\bar{n}}\bar{k}
  =
  \frac{2}{r^2}\left( f-\dfrac{1}{2}rf^\prime \right) ,
\end{eqnarray}
where $\bar n^a$ is the future-directed unit normal to $\sigma_t$ in $S$. 
For the Reissner-Nordstr$\rm{\ddot{o}}$m black hole, it becomes 
\begin{eqnarray}
{}^{(3)}\bar{\mbox \pounds}_{\bar{n}}\bar{k}=\frac{2}{r^2}\Bigl(1-\frac{3m}{r}+\frac{2q^2}{r^2} \Bigr). 
\end{eqnarray}
Thus, ${}^{(3)}\bar{\mbox \pounds}_{\bar{n}}\bar{k}$ is
negative and positive in the inside and outside regions of the photon sphere,
respectively. 
Hence, the non-positivity of
${}^{(3)}\bar{\mbox \pounds}_{\bar{n}}\bar{k}$
is expected to indicate the strong gravity.

We now give the definition of a DTTS \cite{Yoshino:2020-1}. 

\begin{df} \label{def-DTTS}
A closed orientable two-dimensional surface $\sigma_0$ in a smooth spacelike hypersurface $\Sigma$ 
is a dynamically transversely trapping surface (DTTS) if and only if there exists a timelike hypersurface $S$ that intersects 
$\Sigma$ at $\sigma_0$ and satisfies $\bar{k}=0$,
$\mathrm{max}(\bar{K}_{ab}k^ak^b)=0$, and
${}^{(3)}\bar{\mbox \pounds}_{\bar{n}}\bar{k}\le 0$ 
at every point in $\sigma_0$, where $\bar{k}$ is the mean curvature
of $\sigma_0$ in $S$, 
$\bar{K}_{ab}$ is the extrinsic curvature of $S$,
$\bar n^a$ is the unit normal vector of $\sigma_0$ in $S$, 
${}^{(3)} \bar{\mbox \pounds}_{\bar{n}}$ is the Lie derivative associated
with $S$ and $k^a$ is arbitrary future-directed null vectors tangent to $S$. 
\end{df}

Since the emitted photons do not form a photon sphere
in general without spherical symmetry, here we emit photons
in the transverse direction (to satisfy the condition $\bar{k}=0$),
and adopt the location of the outermost photons as the surface $S$
[the condition $\mathrm{max}(\bar{K}_{ab}k^ak^b)=0$].
Then, we judge that the surface $\sigma_0$ exists in a strong
gravity region if the condition
${}^{(3)}\bar{\mbox \pounds}_{\bar{n}}\bar{k}\le 0$
is satisfied. See our previous papers \cite{Yoshino:2020-1, Yoshino:2020-2}
for more details.

%
%

\subsection{Penrose-like inequality for a DTTS}
\label{Sec:PIDTTS}

We present the following theorem
on  a Penrose-like inequality for a DTTS:

\begin{thm} \label{thm-DTTS}
  Consider a spacetime which satisfies the Einstein equations
  with the energy-momentum tensor $T_{ab}$ composed 
  of the Maxwell field part,
  $T_{ab}^{(\rm em)}$ given by Eq.~\eqref{EM-energy-momentum-tensor},
  and the ordinal matter part, $T_{ab}^{(m)}$.
  We suppose that an asymptotically flat spacelike hypersurface $\Sigma$
  is time-symmetric and foliated by the inverse mean curvature flow.
  We also assume that a slice of the foliation parameterized by $y$, where each
  of the $y=\mathrm{constant}$
  surfaces, $\sigma_y$, has topology $S^2$, and $\sigma_0$ is a convex DTTS.   
  We further assume the electric charge density $\rho_{\rm e}$
  to vanish outside $\sigma_0$.
  Then, if $\rho^{(m)}:=T_{ab}^{(m)} n^a n^b \geq 0$
  in the outside region of $\sigma_0$
  and
  $P_r^{(m)}:=T_{ab}^{(m)} r^a r^b \leq 0$
  on $\sigma_0$, 
  where $n^a$ is the future-directed 
  timelike unit normal to $\Sigma$ and
  $r^a$ is the outward spacelike unit normal to $\sigma_y$ in $\Sigma$,   
  the areal radius $r_0=(A_0/4\pi)^{1/2}$ of
  the convex DTTS $\sigma_0$ 
  satisfies the inequality 
\begin{eqnarray}
m \geq \frac{1}{3}\Bigl(1+ \Phi_0^- \Bigr)r_0+\frac{2q^2}{3r_0}, \label{dttsineq}
\end{eqnarray}
where $m$ is the ADM mass and $q$ is the total charge.
Here, $ \Phi^-_0$ is defined by 
\begin{eqnarray}
  \Phi^-_0:=\frac{1}{8\pi}\int_{\sigma_0}[(B_a r^a)^2-(E_aE_b+B_aB_b)h^{ab}]dA,
  \label{Phi-minus-0}
\end{eqnarray}
where $h_{ab}$ is the induced metric of $\sigma_0$. 
\end{thm}

\begin{proof}
On $\sigma_0$, the Lie derivative of the mean curvature $\bar{k}$ along $n^a$ is given by 
\cite{Yoshino:2020-1}
\begin{eqnarray}
{}^{(3)}\bar{\mbox \pounds}_n \bar{k} = -\frac{1}{2}{}^{(2)}R-8\pi P_r+\frac{1}{2}(2kk_{\rm L}+k^2-k_{ab}k^{ab}), 
\label{deriv.mean.dtts}
\end{eqnarray}
where ${}^{(2)}R$ is the Ricci scalar of $\sigma_0$, $k_{ab}$ is the extrinsic curvature of $\sigma_0$ in $\Sigma$, that is, 
$k_{ab}=(1/2){}^{(3)}\mbox \pounds_{r}h_{ab}$, $k$ is its trace, $k_{\rm L}$ is the largest value of the 
eigenvalues of $k_{ab}$ and $P_r$ is the radial pressure defined
in Eq.~\eqref{radial-pressure-general}. 
Then, the condition ${}^{(3)}\bar{\mbox \pounds}_n \bar{k}|_{\sigma_0}\le 0$
gives us 
\begin{equation}
  \frac{3}{4}k^2 \ \le\ \frac12{}^{(2)}R+(E_aE_b+B_aB_b)h^{ab}
  -(E_ar^a)^2-(B_ar^a)^2+8\pi P_r^{(m)}
  \label{Contraction-condition-in-terms-of-E-and-B}
\end{equation}
on $\sigma_0$, where we used Eq.~\eqref{radial-pressure-general} 
and the inequality for a convex DTTS \cite{Yoshino:2020-1}, 
$2kk_{\rm L}+k^2-k_{ab}k^{ab}\ge (3/2)k^2$. 
With the condition $P_r^{(m)}\le 0$,
the integration of the above over $\sigma_0$ gives us
\footnote{Note that the origin for $\Phi^-_0$ in the case of a DTTS
  is different from that for $\Phi^+_0$ in the case of an LTS.
  As we can see from the derivation, $\Phi^\pm_0$ 
  come from the energy density and the radial pressure
  of electromagnetic fields 
  in the cases of the LTS and the DTTS, respectively.
  This is essential reason why we have different results
  for the LTS and the DTTS.}
\begin{eqnarray}
\frac{3}{4}\int_{\sigma_0}k^2dA
&\le&
\frac12\int_{\sigma_0}{}^{(2)}RdA-8\pi \Phi_0^-
-\int_{\sigma_0}(E_ar^a)^2dA\nonumber\\
&\le&
4\pi-\frac{(4\pi q)^2}{A_0}-8\pi \Phi^-_0,\label{k-ineq-dtts}
\end{eqnarray}  
where we used the Gauss-Bonnet theorem, the Cauchy-Schwarz inequality
and Gauss' law as in the proof for Theorem \ref{thm-LTS}.  

We now consider the inverse mean curvature flow in which
the foliation is given by $y=\mathrm{constant}$ surfaces.
Similarly to the proof for Theorem \ref{thm-LTS},
the surfaces $y=0$ and $y=\infty$ are set to be
the DTTS and a sphere at spacelike infinity, respectively.
With the same procedure, 
one can derive the inequality of Eq.~\eqref{monotonic}
again. Then, the integration of the inequality
of Eq.~\eqref{monotonic} over the range $0\le y < \infty$ shows us 
\begin{eqnarray}
  m
  & \geq &
  \frac{A_0^{1/2}}{4\pi^{1/2}}\left(1-\frac{1}{16\pi}\int_{\sigma_0}k^2dA\right)
  +\frac{\pi^{1/2}q^2}{A_0^{1/2}}
  \nonumber\\
  & \geq &
  \frac{A_0^{1/2}}{6\pi^{1/2}}+\frac{4}{3}\frac{\pi^{1/2}q^2}{A_0^{1/2}}+\frac{A_0^{1/2}}{6\pi^{1/2}}{\Phi}^-_0,
\label{DTTSm}
\end{eqnarray}
where we used the inequality of Eq.~\eqref{k-ineq-dtts} at the last step.
Then, we arrive at the inequality given by Eq.~\eqref{dttsineq}. \\
\end{proof}

There are four remarks. 
Similarly to Theorem \ref{thm-LTS}, in general,
the minimum value of the right-hand side of the inequality
of Eq.~\eqref{dttsineq}  
implies us the lower bound for $m$.
However, unlike $\Phi_0^+$ in the case of an LTS,
the quantity $\Phi^-_0$ does not have a definite signature. 
For $\Phi_0^- <-1$, there is no such restriction for $m$,
whereas, for $\Phi_0^- \geq -1$, $m$ has a lower bound, 
\begin{eqnarray}
m \geq \frac{2{\sqrt {2}}}{3}{\sqrt {1+\Phi_0^-}}\ |q|. \label{mbounddtts}
\end{eqnarray}

Next, under the condition of the inequality
of Eq.~\eqref{mbounddtts},
a short rearrangement of the inequality of Eq.~\eqref{DTTSm} gives
\begin{eqnarray}
4\pi r_{{\rm DTTS}-}^2\leq   A_0 \leq 4 \pi r_{{\rm DTTS}+}^2,
  \label{cplidtts}
\end{eqnarray}
where 
\begin{eqnarray}
  r_{{\rm DTTS}\pm}:=\frac{3m \pm \sqrt{9m^2-8\left(1+\Phi^-_0 \right)q^2}}{2\left(1+\Phi^-_0 \right)}.
  \label{RDTTSPM}
\end{eqnarray}
However, from the same reason to the remark in Theorem \ref{thm-LTS},
we expect that the lower bound is not correct 
for a DTTS with multiple components, and just the inequality  
\begin{eqnarray}
A_0 \leq 4\pi r_{{\rm DTTS}+}^2
\end{eqnarray}
would hold in a general context. For a DTTS with a single component,
the lower bound must hold true with the same physical reason as the one
given in Sect.~\ref{Sec:PenroseLike-inequality-LTS}.

As a third remark, in a similar way to the case of an LTS,
the obtained inequality depends on 
the electromagnetic field. Interestingly,
if $\Phi_0^-$ is negative, the contribution
from $\Phi_0^-$ makes the inequality weaker 
than the cases of $\Phi^-_0=0$.
Furthermore, for the case of $\Phi^-_0 \leq -1 $, 
the upper bound disappears. Let us discuss the effect of $\Phi_0^-$
physically. 
It is known that there are two kinds of pressure for magnetic fields.
One is the negative pressure in the direction of magnetic field lines,
called the magnetic tension. The other is
repulsive interaction (i.e., positive pressure)
between two neighboring magnetic field lines, called the magnetic pressure.
A similar thing happens also to electric field lines
(say, the electric tension and the electric pressure).
We recall the formula for $8\pi P_r$,
Eq.~\eqref{radial-pressure-general}.
In that formula, $-(E_ar^a)^2-(B_ar^a)^2$
is the contribution of the electric/magnetic tension, while
$(E_aE_b+B_aB_b)h^{ab}$ is the contribution of the electric/magnetic pressure.
Then, Eq.~\eqref{deriv.mean.dtts} tells that 
the electric/magnetic tension makes the formation of a DTTS difficult,
while the electric/magnetic pressure helps the formation of a DTTS.
Therefore, in the presence of the electric/magnetic pressure,
the area of a DTTS tends to be larger.
This is the reason why the upper bound of the area of a DTTS 
becomes larger when $\Phi_0^-$ is negative.
Nevertheless, the negativity of $\Phi_0^-$ would not change the situation so much in the following reason. 
If $\Phi_0^-$ is negative, the upper bound for the DTTS becomes weaker and the DTTS can exist at farther outside. 
However, $\Phi_0^-$ depends on the position of the DTTS and we naively expect that it is sharply decreasing according to the distance from the center, 
if the electromagnetic field is intrinsic to the compact object; namely, monopole or multi-pole fields. 
Therefore, when we take a farther surface,  $\Phi_0^-$ becomes immediately negligible. 
Then, the area of the DTTS cannot be large. 
On the other hand, $\Phi_0^-$ could be large at some point by extrinsic effects, such as, external fields and/or dynamical generation of fields.  

The final remark is on the relation between
an LTS and a convex DTTS $\sigma_0$ on time-symmetric initial data. Recall 
Proposition 1 in Ref. \cite{Yoshino:2020-1}, that is, a convex DTTS
with $k>0$ in time-symmetric initial 
data is an LTS as well if $\rho+P_r=0$ is satisfied on $\sigma_0$. Since
\begin{eqnarray}
8\pi \left(\rho^{({\rm em})}+P_r^{({\rm em})}\right)=2(E_aE_b+B_aB_b)h^{ab},
\end{eqnarray}
the presence of $\Phi^\pm_0$ disturbs the equivalence between an LTS
and a convex DTTS in general. 
This feature is reflected in the two inequalities obtained in this paper.

%
%

\section{Numerical examination of the Majumdar-Papapetrou spacetime}
\label{Sec:numerical}

In Sect.~\ref{Sec:PIDTTS}, we have obtained the Penrose-like inequality
for a DTTS. There, the quantity $\Phi^-_0$ appears, and this quantity
depends on the configuration of electromagnetic fields.
The purpose of this section is to examine the values 
of $\Phi^-_0$ in an explicit example. Specifically,
in our previous paper \cite{Yoshino:2020-1},
we numerically solved for marginally DTTSs in systems of
two equal-mass black holes adopting the Majumdar-Papapetrou solution.
We revisit this problem from the viewpoint of our current work.

A Majumdar-Papapetrou spacetime is a static
electrovacuum spacetime. The metric is
\begin{equation}
ds^2 = -U^{-2}dt^2 + U^2[d\tilde{r}^2 + \tilde{r}^2(d\tilde{\theta}^2+\sin^2\tilde{\theta}d\tilde{\phi}^2)],
\end{equation}
where the spatial structure is conformally flat, and
we span the spherical-polar coordinates here.
The electromagnetic four-potential is
\begin{equation}
A_a=U^{-1}(dt)_a = -n_a.
\end{equation}
Any solution to the Laplace equation $\bar{\nabla}^2U=0$
gives an exact solution, where $\bar{\nabla}^2$ is the flat space Laplacian.
In this situation, $E^a$ and $B^a$ are calculated as
\begin{equation}
  E^a=-\frac{D^aU}{U}, \qquad B^a=0.
  \label{electric-magnetic-fields-MP}
\end{equation}
Setting $E_\parallel^2:=h_{ab}E^aE^b$,
we have $8\pi \Phi^-_0=-\int_{\sigma_0} E_\parallel^2dA$. 
Therefore, the value of $\Phi^-_0$ is non-positive.

In our previous paper \cite{Yoshino:2020-1}, we chose
the solution
\begin{equation}
  U=1+\frac{m/2}{\sqrt{\tilde{r}^2+z_0^2-2\tilde{r}z_0\cos\tilde{\theta}}}
  + \frac{m/2}{\sqrt{\tilde{r}^2+z_0^2+2\tilde{r}z_0\cos\tilde{\theta}}} 
\end{equation}
that represents the system in which
two extremal black holes with the same charge
are located with the coordinate distance $2z_0$.
Then, assuming the functional form 
$\tilde{r}=h(\tilde{\theta})$, we numerically solved
for a marginally DTTS that surrounds both black holes for each value of $z_0$.
The solution was found in the range
$0\le z_0/m\le 0.79353$. We refer readers to our previous paper
\cite{Yoshino:2020-1} for explicit shapes of the obtained solutions.

%
\begin{figure}[tb]
\centering
\includegraphics[width=0.39\textwidth,bb=0 0 407 300]{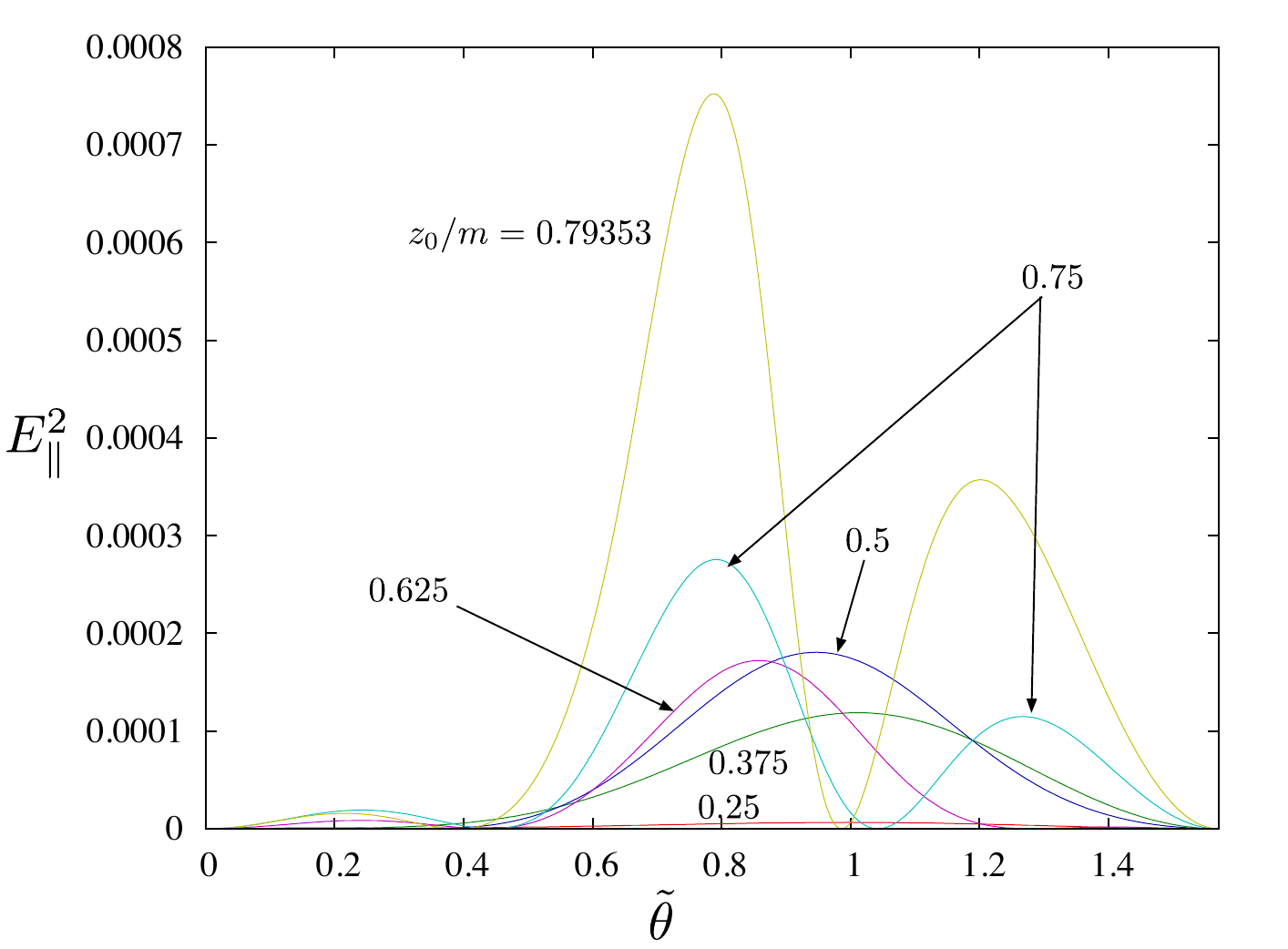}
\hspace{5mm}
\includegraphics[width=0.4\textwidth,bb=0 0 407 301]{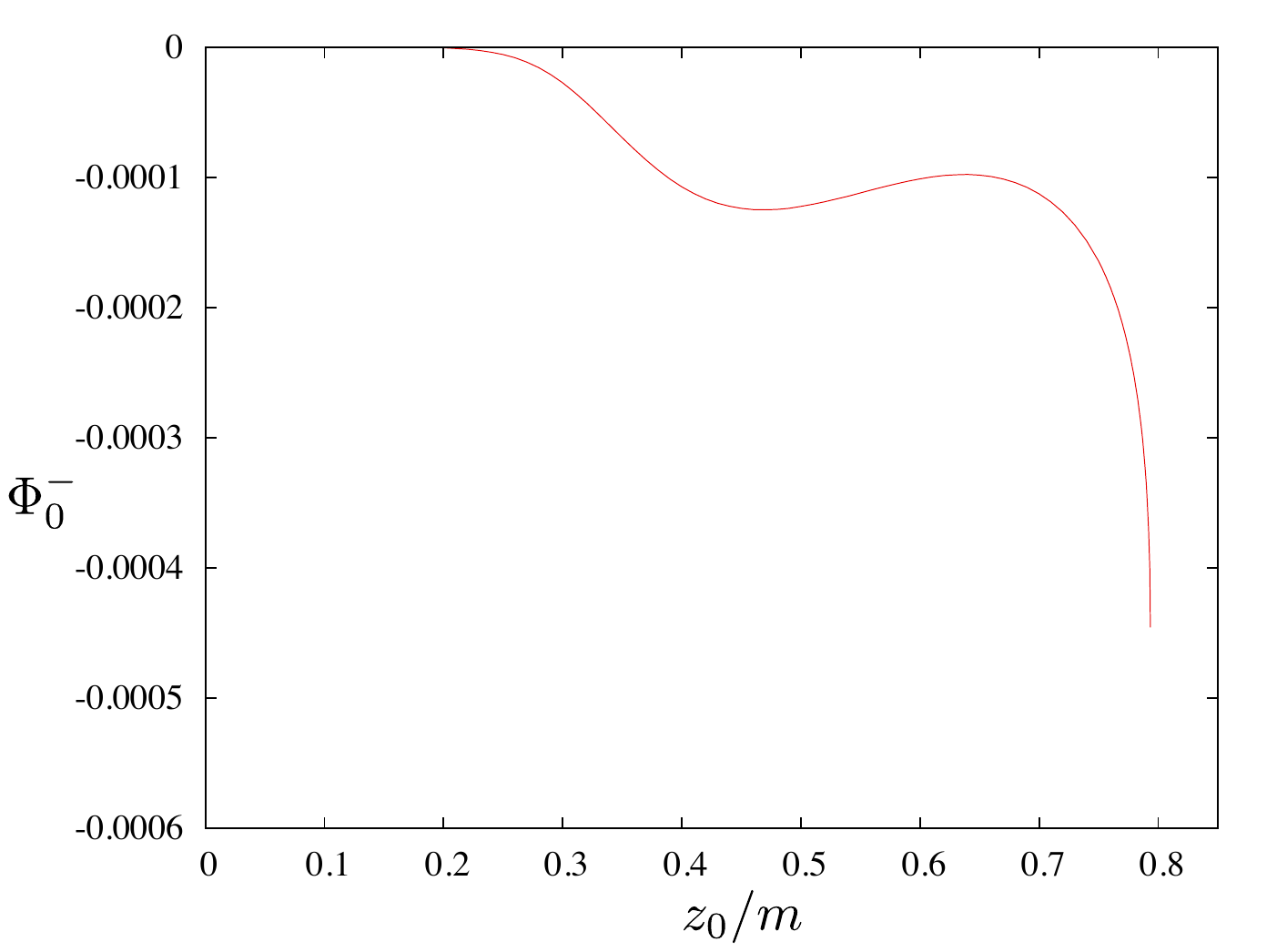}
\caption{The quantities related to $\Phi^-_0$
  of a marginally DTTS in a Majumdar-Papapetrou two-black-hole
  spacetime. Left panel:
  The value of $E_{\parallel}^2$ as a function of
  $\tilde{\theta}$ on a marginally DTTS
  for $z_0/m = 0.25$, $0.375$, $0.5$, $0.625$, $0.75$,
  and $0.79353$.
  Right panel: The value of $\Phi^-_0$ for a marginally DTTS
  as a function of $z_0/m$. 
}
\label{Epara-theta-combined}
\end{figure}
%

We examine the value of $\Phi^-_0$. 
The left panel of Fig.~\ref{Epara-theta-combined} presents
the behavior of $E_\parallel^2$ as a function of $\tilde{\theta}$
on a marginally DTTS for $z_0/m = 0.25$, $0.375$, $0.5$, $0.625$, $0.75$,
and $0.79353$. The value of $E_\parallel^2$ is generally nonzero,
but is less than $10^{-3}$. 
The right panel of Fig.~\ref{Epara-theta-combined}
plots the value of $\Phi^-_0$ as a function of $z_0$.
It is negative and its absolute value is less than $5\times 10^{-4}$.

%
\begin{figure}[tb]
\centering
\includegraphics[width=0.4\textwidth,bb= 0 0 407 298]{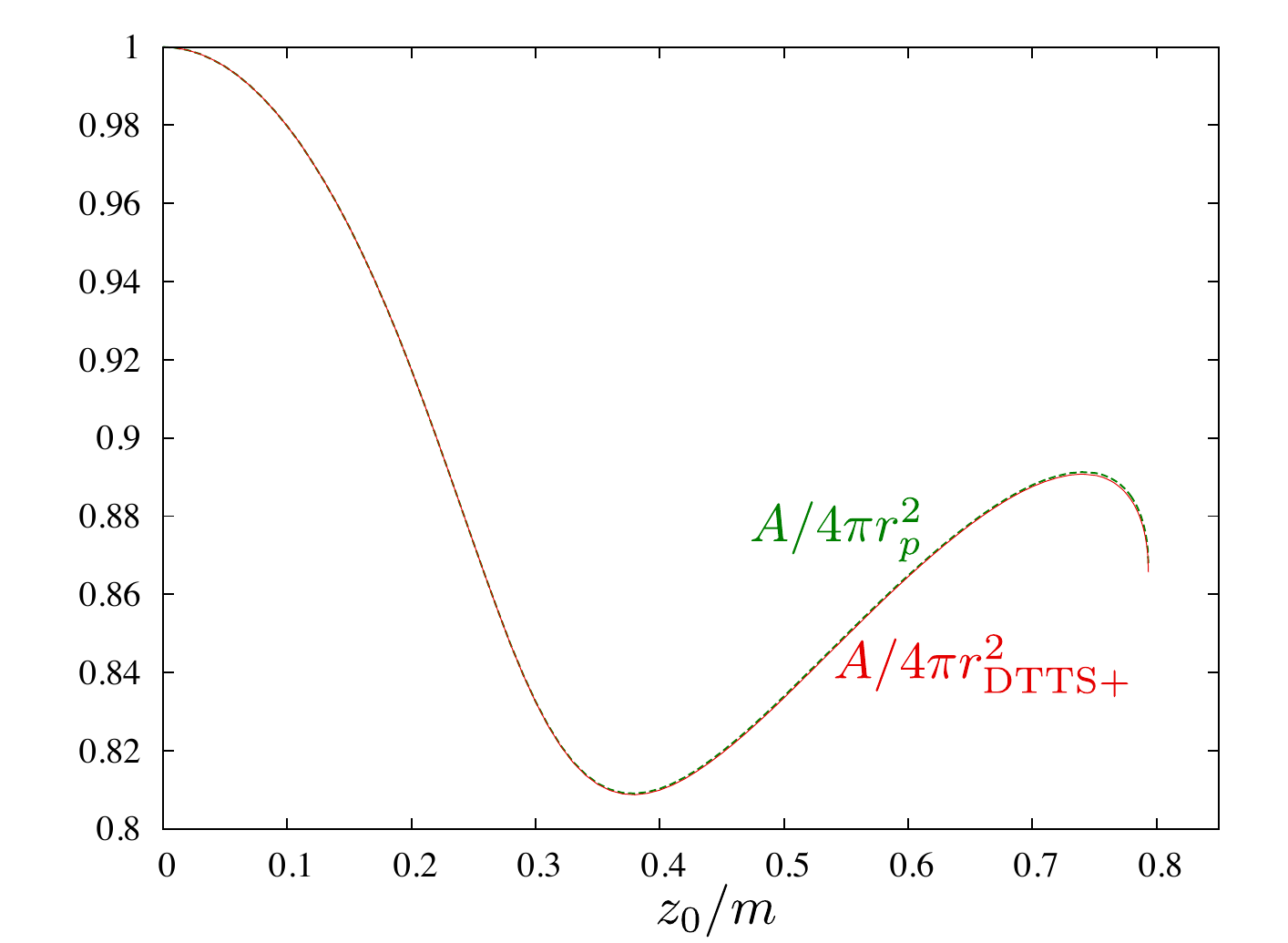}
\caption{The area $A$ of a common marginally DTTS
  in a Majumdar-Papapetrou spacetime as a function of $z_0/m$.
  The cases of two kinds of normalizations are shown.
  One is $A/4\pi r_{\rm DTTS+}^2$ and the other is $A/4\pi r_p^2$.
  See text for details.
}
\label{area-z0-combined}
\end{figure}
%

Figure~\ref{area-z0-combined} shows the relation
between the area $A$ of the marginally DTTS and $z_0$.
We normalize the value of $A$ in two ways: One is
$A/4\pi r_{\rm DTTS+}^2$ (a red solid curve), where $r_{\rm DTTS+}$ is defined in
Eq.~\eqref{RDTTSPM}, and the other is $A/4\pi r_p^2$
(a green dotted curve), 
where $4\pi r_p^2$ is the area of a photon sphere
with the same mass and charge [see Eq.~\eqref{rp}]. 
Because the value of $\Phi^-_0$ is small, the difference
is scarcely visible. Both of these values are
in agreement with the Penrose-like inequalities.

%
\begin{figure}[tb]
\centering
\includegraphics[width=0.4\textwidth,bb=0 0 302 300]{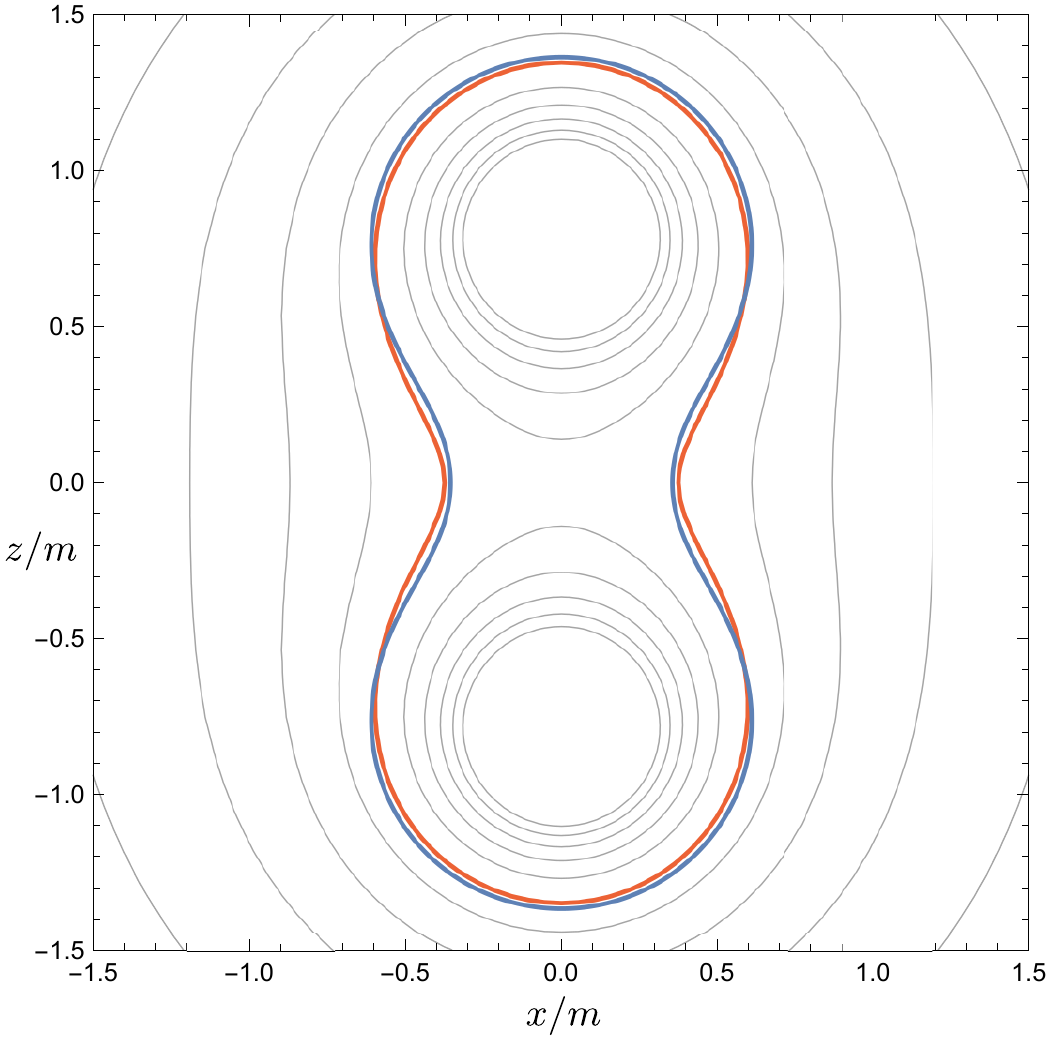}
\caption{The marginally DTTS for $z_0/m=0.79353$
  (the blue curve)
  and the contour surfaces of $U$ in the Majumdar-Papapetrou two-black-hole
  spacetime (gray and red curves). The red curve depicts the
  contour surface of $U=2.14$. 
}
\label{Compare-Contour-DTTS}
\end{figure}
%

The reason why the value of $\Phi^-_0$ is so small
is that the electric field is approximately perpendicular
to the marginally DTTS. From the formula for $E^a$ given in
Eq.~\eqref{electric-magnetic-fields-MP},
this means that the marginally DTTS approximately coincides
with a contour surface of $U$. Figure~\ref{Compare-Contour-DTTS}
confirms this feature for $z_0/m=0.79353$.
Here, the blue curve depicts the marginally DTTS, and the red curve
shows the contour surface of $U=2.14$. They agree well.

The lesson from this numerical experiment is that
if a spacetime is static, the quantity $\Phi^-_0$ is small
and does not play an important role in the Penrose-like inequality
for a DTTS.
Of course, it is expected that the absolute value of $\Phi^-_0$
may become large if dynamical situations are considered.
For example, if two black holes have opposite charges,
the contribution of the electric pressure would become important,
although such a situation is more difficult to study.
Exploring such issues is left as a remaining problem.

%
%

\section{Summary and discussion}
\label{Sec:Sum}

In this paper, we have examined
the properties of LTSs and DTTSs
for Einstein-Maxwell systems, particularly focusing
on the derivation of Penrose-like inequalities on their area.
Similarly to the Riemannian Penrose inequality for 
charged cases, the electric charge comes into the inequalities,
but there are additional contributions from the density or 
the pressure/tension of electromagnetic fields in general.
This is a rather interesting result because 
one naively expects that
the upper bound for the area of an LTS and a DTTS is 
that of the photon sphere $4\pi r_p^2$, where 
$r_p$ is the radius 
of an unstable circular orbit of a photon
in the Reissner-Nordstr\"om spacetime given in Eq.~\eqref{rp}.  
For an LTS, we have a tighter inequality than the naive one.
For a DTTS, the obtained inequality can become
both stronger and weaker depending on the configuration
of electromagnetic fields.
We have numerically examined the value of the correction term,
represented by $\Phi^-_0$
in Eq.~\eqref{Phi-minus-0},
for a Majumdar-Papapetrou two-black-hole spacetime.
Although the correction term makes the inequality weaker,
we have checked that the value of $\Phi^-_0$ is very small
in that situation.

Up to now, we have assumed
that the magnetic charge $q_{\rm m}$ vanishes.
Here, we consider what happens to our inequalities
when nonzero $q_{\rm m}$ is present.
We first consider the case of an LTS. Since
the total squared charge $q^2 := q_{\rm e}^2+q_{\rm m}^2$ appears in the metric
of the Reissner-Nordstr\"om solution when both the electric and
magnetic fields are present,
we would like to present the Penrose-like inequality
in terms of $q^2$. For this reason, we have to
use the inequality
\begin{equation}
  \int_{S_0} \left[(E_ar^a)^2 +(B_ar^a)^2\right]dA\ge
  A_0^{-1} \left[(4\pi q_{\rm e})^2+(4\pi q_{\rm m})^2\right] = \frac{(4\pi q)^2}{A_0}
  \label{Cauchy-Schwarz-in-the-case-of-nonzero-qm}
\end{equation}
in the calculations of Eqs.~\eqref{calculation-2DRicci-LTS} and \eqref{k-ineq}.
As a result, we must introduce the quantity
\begin{equation}
  \Phi_0 := \frac{1}{8\pi}\int_{S_0}\left[(E_aE_b+B_aB_b)h^{ab}\right]dA
  \label{Def:Phi-0}
\end{equation}
instead of $\Phi^+_0$ of Eq.~\eqref{Phi-plus-0}.
The resultant inequality is the one of Eq.~\eqref{ltsineq}
but $\Phi^+_0$ being replaced by $\Phi_0$.
Next, we consider the case of a DTTS. Similarly to the case of an LTS,
we must use the inequality of Eq.~\eqref{Cauchy-Schwarz-in-the-case-of-nonzero-qm}
(but $S_0$ being replaced by $\sigma_0$) in the calculations
of Eqs.~\eqref{Contraction-condition-in-terms-of-E-and-B} and \eqref{k-ineq-dtts}.
As a result, instead of $\Phi_0^-$, we must introduce $-\Phi_0=
-(1/8\pi)\int_{\sigma_0}\left[(E_aE_b+B_aB_b)h^{ab}\right]dA$.
The resultant inequality is the same as the one of Eq.~\eqref{dttsineq},
but $\Phi^-_0$ being replaced by $-\Phi_0$. These results are summarized
as follows:
\begin{cor}
  In the presence of nonzero $q_{\rm m}$, Theorems~\ref{thm-LTS} and
  \ref{thm-DTTS} hold by changing from  $\Phi^\pm_0$ to $\pm \Phi_0$,
  where $\Phi_0$ is defined by Eq.~\eqref{Def:Phi-0}.
\end{cor}

In the main article of this paper, we have not considered a TTS 
for the static and stationary spacetimes
defined in our previous paper \cite{Yoshino:2017}.
We note that the concepts of a TTS and a DTTS are
related but independent of each other
in the sense that no inclusion relationship can be found \cite{Yoshino:2020-1}.
In Appendix~\ref{Appendix-A}, we present a theorem on the Penrose-like
inequality for a TTS in a static spacetime, which is very similar
to Theorem~\ref{thm-DTTS}.

Throughout this study, we have not used the property of
Maxwell's equations except for Gauss' law.
The information from Maxwell's equation may further restrict the
properties of LTSs, DTTSs, and TTSs, especially for
static/stationary spacetimes with static/stationary electromagnetic fields. 

%

\ack

T. S. and K. I.  are supported by Grant-Aid for Scientific Research from Ministry of Education, 
Science, Sports and Culture of Japan (No. 17H01091). 
H.Y. is supported by the Grant-in-Aid for
Scientific Research (C) (No. JP18K03654) from Japan Society for
the Promotion of Science (JSPS).
The work of H.Y. is partly supported by
Osaka City University Advanced Mathematical Institute
(MEXT Joint Usage/Research Center on Mathematics and Theoretical Physics).

\appendix

%
%
\section{Transversely trapping surface}
\label{Appendix-A}

The four of the present authors also proposed the concept of a
TTS \cite{Yoshino:2017}.
This concept is applicable only to static or stationary spacetimes. 
The definition is as follows :

\begin{df} \label{def-TTS}
A static/stationary timelike hypersurface $S$ is a transversely trapping surface (TTS) if and 
only if arbitrary light rays emitted in arbitrary tangential directions of $S$ from arbitrary 
points of $S$ propagate on $S$ or toward the inside region of $S$.
\end{df}

The necessary and sufficient condition for a surface $S$ to be a TTS
(the TTS condition hereafter)
is expressed as $\bar{K}_{ab}k^ak^b\le 0$, where $\bar{K}_{ab}$
is the extrinsic curvature of $S$ and $k^a$ are arbitrary null vectors
tangent to $S$. For a static spacetime, there is 
the Killing time coordinate $t$ whose basis vector $t^a$
is orthogonal to the $t=\mathrm{constant}$ hypersurface,
denoted by $\Sigma$. 
The lapse function $\alpha$ is defined by $t^a=\alpha n^a$,
where $n^a$ is a future-directed unit normal to $\Sigma$.
We denote the two-dimensional
section of $S$ and $\Sigma$ by $\sigma_0$. 
As shown in Eq.~(15) of Ref.~\cite{Yoshino:2017}, 
the TTS condition is reexpressed in terms of $\alpha$ as
\begin{equation}
k_{\rm L}\le\frac{r^aD_a\alpha}{\alpha},
\end{equation}
where $k_{\rm L}$ is the largest value among the two eigenvalues
of the extrinsic curvature $k_{ab}$ of $\sigma_0$ in $\Sigma$,
$r_a$ is a spacelike unit normal to $S$, and $D_a$ is the
covariant derivative with respect to $\Sigma$. 
In this situation, it is possible to
derive the relation
\begin{equation}
  {}^{(2)}R=-16\pi P_r+\frac{2}{\alpha}\mathcal{D}^2\alpha
  +2k\frac{r^aD_a\alpha}{\alpha} + k^2-k_{ab}k^{ab}.
\end{equation}
For a convex TTS,
we can derive the inequality
\begin{equation}
2k\frac{r^aD_a\alpha}{\alpha}+k^2-k_{ab}k^{ab}\ge \frac32k^2.
\end{equation}
These relations are presented as Eqs.~(24) and (25) in Ref.~\cite{Yoshino:2017}.
We now consider the Einstein-Maxwell system. 
Using Eq.~\eqref{radial-pressure-general}, we find
\begin{equation}
  \frac{3}{4}k^2 \ \le\ \frac12{}^{(2)}R+(E_aE_b+B_aB_b)h^{ab}
  -(E_ar^a)^2-(B_ar^a)^2+8\pi P_r^{(m)}
  -\frac{1}{\alpha}\mathcal{D}^2\alpha.
\end{equation}
Compare this inequality with the one of Eq.~\eqref{Contraction-condition-in-terms-of-E-and-B}.
Integrating over $\sigma_0$, we obtain exactly the same
inequality as the one of Eq.~\eqref{k-ineq-dtts}.
Therefore, a TTS in a static spacetime satisfies the same inequality
as the Penrose-like inequality for a DTTS in time-symmetric initial data.
This result is summarized as the following theorem:

\begin{thm} \label{thm-TTS}
  The static time cross section of a convex TTS,
  $\sigma_0$, in an asymptotically flat static spacetime has topology $S^2$
  and its areal radius $r_0=\sqrt{A_0/4\pi}$  
  satisfies the inequality of Eq.~\eqref{dttsineq}
  [with $\Phi^-_0$ defined in Eq.~\eqref{Phi-minus-0}] 
  if $P_{r}^{(m)}<0$ holds on $\sigma_0$, $k>0$ at least 
  at one point on $\sigma_0$, and $\rho^{(m)} \geq 0$
  in the outside region.
\end{thm}

As remarked in the final section, this theorem applies to the case
that the magnetic charge $q_{\rm m}$ is zero.
When $q_{\rm m}$ is nonzero, $\Phi^-_0$ must be replaced by $-\Phi_0$,
where $\Phi_0$ is given in Eq.~\eqref{Def:Phi-0}.


\end{document}